\documentclass[12pt]{article}
\usepackage{graphicx,float,hyperref, indentfirst} 
\usepackage{amsmath,amsthm,amssymb,amsfonts,geometry,mathtools,xcolor, enumerate}
\usepackage{subfigure}
\usepackage{tikz}
\usepackage[noend]{algpseudocode}
\usepackage{algorithm}

\theoremstyle{plain}
\newtheorem{theorem}{Theorem}[section]
\newtheorem{lemma}[theorem]{Lemma}

\theoremstyle{remark}
\newtheorem{remark}{Remark}

\newcommand{\Bin}{\ensuremath{\textrm{Bin}}}

\newcommand{\E}{\mathbb{E}}

\newcommand{\supp}{\text{supp }}

\title{A quantum algorithm for learning a graph of bounded degree}
\author{Asaf Ferber \thanks{Department of Mathematics, University of California, Irvine.
Email: \href{mailto:asaff@uci.edu} {\nolinkurl{asaff@uci.edu}}.
Research supported in part by NSF Awards DMS-1954395 and DMS-1953799, and by an AFOSR grant FA9550-23-1-0298.}\and Liam Hardiman \thanks{Department of Mathematics, University of California, Irvine.
Email: \href{mailto:lhardima@uci.edu} {\nolinkurl{lhardima@uci.edu}}.}}
\date{\today}

\begin{document}

\maketitle

\begin{abstract}
    We are presented with a graph, $G$, on $n$ vertices with $m$ edges whose edge set is unknown.
    Our goal is to learn the edges of $G$ with as few queries to an oracle as possible.
    When we submit a set $S$ of vertices to the oracle, it tells us whether or not $S$ induces at least one edge in $G$.
    This so-called OR-query model has been well studied, with Angluin and Chen giving an upper bound on the number of queries needed of $O(m \log n)$ for a general graph $G$ with $m$ edges.
    
    When we allow ourselves to make *quantum* queries (we may query subsets in superposition), then we can achieve speedups over the best possible classical algorithms.
    In the case where $G$ has maximum degree $d$ and is $O(1)$-colorable, Montanaro and Shao presented an algorithm that learns the edges of $G$ in at most $\tilde{O}(d^2m^{3/4})$ quantum queries.
    This gives an upper bound of $\tilde{O}(m^{3/4})$ quantum queries when $G$ is a matching or a Hamiltonian cycle, which is far away from the lower bound of $\Omega(\sqrt{m})$ queries given by Ambainis and Montanaro.
    
    We improve on the work of Montanaro and Shao in the case where $G$ has bounded degree. In particular, we present a randomized algorithm that, with high probability, learns cycles and matchings in $\tilde{O}(\sqrt{m})$ quantum queries, matching the theoretical lower bound up to logarithmic factors.
\end{abstract}

\section{Introduction}

Quantum algorithms are able to solve certain computational problems in fewer steps than any classical algorithm.
For example, if a step consists of comparing an item in an unsorted list of length $N$ against some fixed $x_0$, then any classical algorithm will require $N$ steps in the worst case to pick $x_0$ out of the list, while Grover \cite{Grover} showed that a quantum algorithm can accomplish the same task in $O(\sqrt N)$ steps, with high probability (a matching lower bound was shown earlier in \cite{bennett1997strengths}).
If a step is evaluating a function, promised to be a dot product modulo 2 with a fixed $n$ bit string $s$, then elementary linear algebra shows that we need $n$ steps to determine $s$, while Bernstein and Vazirani \cite{bernstein1993quantum} showed that a single step suffices for a particular quantum algorithm.
In both of these cases, the savings come from the fact that a single step, or query in these cases, on a quantum computer may consist of many classical steps taken in superposition.
There is a catch however, as we (usually) cannot directly access the results of the individual queries that we made in superposition since measuring a quantum state is a random process that destroys information.
These quantum algorithms work by making queries and then manipulating the state in such a way that the final measurement yields a solution (perhaps with some probability of error).

Suppose that we want to learn the edge set of some graph $G$ that is promised to have $n$ vertices and $m$ edges.
This problem is not tremendously interesting if our queries are to the adjacency matrix of $G$, i.e., we may ask an oracle whether a pair of vertices spans an edge.
This amounts to learning an $\binom{n}{2}$-bit string one bit at a time, which clearly requires $\Omega(n^2)$ queries in the worst case, even if we allow ourselves to use a randomized algorithm that errs, say, a third of the time.
Repeated application of Grover's algorithm, on the other hand, can find all $m$ edges in $O(n\sqrt{m})$ quantum queries (with high probability).
This does amount to a significant speedup, but only when $G$ is sparse.
If $m$ is not given, then we require $\Omega(n^2)$ queries for both classical and quantum algorithms.
The fact that a quantum algorithm can do no better in this case (up to the constant factor implicit in the asymptotic notation) follows from a deeper result of Beals et al. \cite{beals2001quantum} that a quantum algorithm requires $\Omega(N)$ queries to learn even the \emph{parity} of an $N$-bit string.

Here we consider a more general model, where we may query an arbitrary subset of vertices $S\subseteq V(G)$ and ask whether or not $(S\times S) \cap E$ is empty.
This so-called OR query model has been well-studied and has applications to gene sequencing \cite{beigel2001optimal, bouvel2005combinatorial}.
Alon et al. \cite{alon2004learning} showed the importance of the adaptiveness (whether or not the sequence of queries needs to be fixed in advance) of an algorithm that learns a matching in this model.
They showed that a deterministic nonadaptive algorithm requires $\Omega(n^2)$ queries, while a randomized adaptive algorithm requires just $\Theta(n \log n)$ queries.
In the case of a general graph on $m$ edges, Angluin and Chen \cite{AngluinChen} present an adaptive algorithm that learns the edge set in $O(m\log n)$ OR queries.

As for quantum algorithms, D{\"u}rr et al. \cite{durr2006quantum} discuss the quantum query complexity of graph properties (like connectivity) in the adjacency matrix and adjacency list models, while Lee, Santha and Zhang \cite{lee2021quantum} investigate graph learning using cut queries (querying a subset $S$ returns the number of edges with exactly one end in $S$).
Ambainis and Montanaro \cite{ambainisMontanaro} give a lower bound for a quantum algorithm that solves the combinatorial group testing problem, which we will discuss in the next section.
Their result implies a lower bound of $\Omega(\sqrt m)$ quantum OR queries for learning a graph with $m$ edges.
Montanaro and Shao \cite{montanaroShao} give an explicit quantum algorithm, and thereby, an upper bound on the number of quantum OR queries needed to learn an unknown graph.
\begin{theorem}[\cite{montanaroShao} Theorem 7]\label{thm: montanaro shao}
    Let $G$ be a graph on $n$ vertices with $m$ unknown edges.
    Then there is a quantum algorithm that identifies all $m$ edges in $G$ using at most
    \[
    O\left( m\log(\sqrt m\log n) + \sqrt{m}\log n \right)
    \]
    quantum OR queries with probability at least 0.99.
    If $G$ has maximum degree $d$ and is $O(1)$-colorable, then at most
    \[
    O\left( d^2m^{3/4}\sqrt{\log n}(\log m) + \sqrt{m}\log n \right)
    \]
    quantum OR queries are necessary with probability at least 0.99.
\end{theorem}

In particular, their algorithm learns matchings and Hamiltonian cycles in $\tilde{O}(m^{3/4})$ queries (where the $\tilde{O}$ omits factors logarithmic in $n$).
Our main contribution is the following pair of theorems.
The first is a special case of the second, but its proof is slightly simpler.

\begin{theorem}\label{thm: matching}
    Let $G$ be a matching on $n$ vertices with $m$ unknown edges.
    Then there is a quantum algorithm that identifies all $m$ edges in $G$ using at most 
    \[ 
    O\left( \sqrt{m}\log m \log^{11/2}n\right) = \tilde{O}(\sqrt m)
    \]
    quantum OR queries, where the $\tilde O$ omits factors logarithmic in $n$.
    The algorithm succeeds with probability $1-o(n^{-3/2})$. 
\end{theorem}

\begin{theorem}\label{thm: max degree}
    Let $G$ be a graph on $n$ vertices with $m$ unknown edges and maximum degree $d$.
    Then there is a quantum algorithm that identifies all $m$ edges in $G$ using at most 
    \[ 
    O\left( d\sqrt{m}\log(m/d)\log^{11/2}n\right) = \tilde{O}(d\sqrt{m})
    \]
    quantum OR queries.
    The algorithm succeeds with probability $1-o(n^{-1/2})$. 
\end{theorem}

Up to logarithmic factors, our results match Ambainis and Montanaro's lower bound of $\Omega(\sqrt m)$ quantum queries for learning a graph with maximum degree $d = O(1)$.

\section{Notation and tools}
All graphs in this paper are assumed to be simple.
If $G$ is a graph and $U$ is a subset of the vertex set, $V(G)$, then $E(U)$ is the set of edges with both ends in $U$ and $e(U) = |E(U)|$.
Similarly, if $W$ is another subset of $V(G)$, then $E(U, W)$ is the set of edges with one end in $U$ and the other in $W$, and $e(U, W) = |E(U, W)|$.
We will make use of the elementary result by Chernoff, bounding the lower and the upper tails of the binomial distribution (see \cite{JRLrandomgraphs}).
	\begin{theorem}[Chernoff bound]\label{Chernoff}
        Let $X_1, \ldots, X_n$ be independent Bernoulli random variables and let $X  = X_1 + \cdots +X_n$.
        If $\mu = \E[X]$, then for every $a>1$,
        \begin{itemize}
			\item $\Pr[X < \mu/a]<\left(\frac{e^{1/a - 1}}{a^a}\right)^\mu$;
			\item $\Pr[X>a\mu] < \left(\frac{e^{a-1}}{a^a}\right)^\mu$.
		\end{itemize}
        The following simpler bound also holds for any $\delta\geq 0$
        \begin{itemize}
            \item $\Pr[X > (1+\delta)\mu] \leq e^{-\frac{\delta^2}{2+\delta}\mu}$
        \end{itemize}
        or for any $0<\delta<1$
        \begin{itemize}
            \item $\Pr[|X-\mu| > \delta\mu] \leq 2e^{-\frac{\delta^2}{3}\mu}$
        \end{itemize}
	\end{theorem}
	
	\begin{remark}[see e.g. \cite{JansonConcentration}]
		The above bounds also hold when the $X_i$'s are \emph{negatively correlated}, i.e. when the $X_i$'s are such that for any subset $I\subseteq [n]$,
          \[
            \E\left[\prod_{i\in I}X_i\right] \leq \prod_{i\in I}\E[X_i].
          \]
        For example, the Chernoff bound holds when $X$ is a hypergeometric random variable.
	\end{remark}

On the topic of binomial random variables, we use the following to estimate the probability that such a variable attains its mean.
The proof immediately follows from Stirling's approximation.
\begin{lemma}\label{lem: binomial}
    Let $n$ be a positive integer and $p = p(n)\in [0, 1]$ be such that $np$ and $n(1-p)$ both tend to infinity as $n\to \infty$.
    Suppose $X \sim \Bin(n,p)$. Then
    \[
    \Pr[X = np] = \binom{n}{np}p^{np}(1-p)^{n(1-p)} \sim \frac{1}{\sqrt{2\pi np(1-p)}}.
    \]
\end{lemma}

As for graph theoretic tools, we use Vizing's classical result on the edge chromatic number.
For any graph $G$ and positive integer $k$, an \emph{edge-coloring} of $G$ is a map $c: E(G)\to [k]$ and we say that $c$ is \emph{proper} if $c(e) \neq c(f)$ whenever the distinct edges $e$ and $f$ intersect.
The \emph{chromatic index} of $G$, $\chi'(G)$, is the smallest positive integer $k$ such that $G$ admits a proper coloring $c: E(G)\to [k]$.
We clearly have $\chi'(G) \geq \Delta$ (as all edges incident to a vertex of maximum degree must have different colors), and odd cycles show that we might need $\Delta+1$ colors.
Vizing's theorem says that we will never need more than $\Delta+1$ colors.

\begin{theorem}[Vizing's Theorem \cite{Vizing}]\label{Vizing}
    Suppose $G$ is a graph with maximum degree $\Delta$.
    Then $\chi'(G)$ is either $\Delta$ or $\Delta+1$.
\end{theorem}

Moving onto algorithmic tools, we will make use of Angluin and Chen's classical algorithm for learning a graph with OR queries.
\begin{theorem}[\cite{AngluinChen} Theorem 3.1]\label{thm: classical}
    Let $G$ be a graph with $n$ vertices and $m$ edges.
    Then there is a (classical) deterministic adaptive algorithm that identifies all $m$ edges in $G$ using at most $O(m \log n)$ OR-queries.
\end{theorem}

One of the key subroutines in the algorithm of Montanaro and Shao is an efficient quantum algorithm for \emph{combinatorial group testing} (CGT).
An instance of this problem consists of an $n$-bit string $S$ with Hamming weight $k$.
We are allowed to query arbitrary substrings of $S$ and receive their bitwise OR as output and our goal is to determine $S$ with as few queries as possible.
Classically, CGT can be solved by an adaptive algorithm in $O(\log \binom nk)$ queries using, e.g. Hwang's binary splitting algorithm \cite{du2000combinatorial} (see chapter 2).
This matches the information-theoretic lower bound.

Ambainis and Montanaro \cite{ambainisMontanaro} gave a quantum algorithm for CGT that uses $O(k\log k)$ queries, beating the classical lower bound of $\Omega(\log \binom nk) = \Omega(k\log (n/k))$ when $k = O(\sqrt n)$.
They also showed that $\Omega(\sqrt k)$ quantum queries are necessary, and this was shown by Belovs \cite{belovs}, using the adversary method, to be tight.

\begin{theorem}[\cite{belovs} Theorem 3]\label{thm: cgt}
    The quantum query complexity of the combinatorial group testing problem is $\Theta(\sqrt k)$.
    That is, any quantum algorithm that solves the CGT problem requires $\Omega(\sqrt k)$ queries and there is an algorithm that uses $O(\sqrt k)$ queries.
\end{theorem}

\begin{remark}
    Belovs' algorithm for solving the CGT problem arises from solving a semidefinite program.
    While this does yield an explicit algorithm, it might not be \emph{time} efficient, i.e., we might only need $O(\sqrt k)$ queries, but we might need to use the results as input to some hard computational problems.
    Montanaro and Shao \cite{montanaroShao} give an explicit time efficient algorithm for CGT that uses $O(\sqrt{k}\log k \log\log k)$ queries.
    This algorithm is based of the Bernstein-Vazirani \cite{bernstein1993quantum} algorithm for learning a binary linear functional and Ambainis et al.'s \cite{ambainis2016efficient} algorithm for gapped group testing.
\end{remark}

\begin{remark}
    If $G$ is a graph on $n$ vertices with $m$ edges, then we may represent its edge set as an $\binom{n}{2}$-bit string, $S$, with weight $m$.
    The task of learning the edges of $G$ by OR-querying subsets of vertices is equivalent to a restricted version of the CGT problem for learning $S$, where we are only allowed to query certain substrings of $S$.
    Specifically, if a query to $S$ includes the positions corresponding to pairs $\{u, v\}$ and $\{u',v'\}$, then it must also include those corresponding to $\{u, u'\}$, $\{u, v'\}$ and $\{u', v\}$.
    Consequently, Ambainis and Montanaro's lower bound of $\Omega(\sqrt m)$ quantum queries to learn $S$ with CGT gives a lower bound of $\Omega(\sqrt m)$ quantum OR-queries for learning the edge set of $G$.
\end{remark}

We use this algorithm for CGT to develop a randomized algorithm for finding edges that cross between independent sets.
Our algorithm is inspired by that of Montanaro and Shao \cite{montanaroShao}, itself inspired by Angluin and Chen's \cite{AngluinChen} classical algorithm.

\begin{lemma}\label{lem: improved bipartite case}
    Let $G$ be a graph on $n$ vertices with maximum degree $d$.
    Suppose $A$ and $B$ are two disjoint, non-empty independent sets of vertices in $G$, and that there are $m_{AB}$ edges between $A$ and $B$.
    Then there is a quantum algorithm that identifies the $m_{AB}$ crossing edges in $O(d\sqrt{m_{AB}}\log n)$ queries.
    The algorithm succeeds with probability $1-o(n^{-4})$.
\end{lemma}
\begin{proof}
    For any subset $T\subseteq B$, querying sets of the form $S\cup T$, where $S$ ranges over the subsets of $A$, reduces the problem of learning $N(T)$ to the CGT problem.
    In particular, if there are $n_A$ non-isolated vertices in $A$ and $n_B$ non-isolated vertices in $B$, then we can learn the non-isolated vertices in $O(\sqrt{n_A} + \sqrt{n_B})$ queries.
    We therefore assume that there are no isolated vertices in $A$ or $B$.
    While Theorem \ref{thm: cgt} allows us to learn $N(T)$ in $O(\sqrt{|N(T)|})$ queries, we do not learn the individual adjacencies (i.e., if $a\in N(T)$, we do not learn the identities of its neighbors in $T$).

    Set $N = 60 d\log n$ and let $T_1, \ldots, T_N$ be random subsets of $B$, where each $b\in B$ appears in $T_i$ with probability $p = \frac{1}{3d}$ independently for all $i$.
    Notice that, by Chernoff and a union bound, each $b\in B$ appears in $\Theta(\log n)$ of the $T_i$'s with probability, say, $1-n^{-4}$.
    Now for every $a\in A$ and $b\in B$, define the indicator vectors $\chi(a),\chi(b)\in \{0,1\}^N$ by
    \[
    \chi(a)_i = \begin{cases}
        1,&\text{if }a\in N(T_i)\\
        0,& \text{otherwise}
    \end{cases},
    \qquad
    \chi(b)_i = \begin{cases}
        1,&\text{if }b\in T_i\\
        0, &\text{otherwise}
    \end{cases}.
    \]
    We denote the support of a binary vector $u\in \{0,1\}^N$ by $\supp u:= \{i\in [N]: u_i = 1\}$ and its size by $|u|$.
    Choosing the $T_i$'s determines the $\chi(b)$'s and learning the $N(T_i)$'s with the CGT algorithm gives us the $\chi(a)$'s.
    Notice that if $a$ has neighbors $b_1, \ldots, b_t$, then
    \[
    \chi(a) = \bigvee_{i=1}^t\chi(b_i),
    \]
    where $u \lor v$ is the bitwise OR of the binary vectors $u$ and $v$.
    Consequently, if $b$ is a neighbor of $a$, then $\supp \chi(b)\subseteq \supp \chi(a)$.

    Fix some $a\in A$ and let $b\in B$ be some non-neighbor of $a$.
    We claim that, with high probability, $\supp \chi(b) \not\subseteq \supp \chi(a)$.
    Indeed, if $|N(a)| = t$, then for any $i$, $\Pr[a\in N(T_i)] =1-(1-p)^t$, so
    \[
    \E |\chi(a)| = N[1-(1-p)^t] \leq N[1-(1-p)^d] \leq Npd = \frac{1}{3}N.
    \]
    Because we may assume that there are no isolated vertices, we also have the lower bound $\E|\chi(a)| \geq Np$.
    Since $|\chi(a)|$ is a sum of independent Bernoulli random variables, the Chernoff bound tells us that
    \[
    \Pr\bigg[|\chi(a)| \geq \frac{2}{3}N\bigg] \leq e^{-\frac{1}{3}\E|\chi(a)|} \leq e^{-Np/3} = o\left(n^{-20/3}\right).
    \]
    Now $\supp \chi(b) \subseteq \supp \chi(a)$ if and only if $b\notin T_i$ for every $i\notin \supp \chi(a)$, so
    \begin{align*}
        \Pr[\supp \chi(b)\subseteq \supp \chi(a)] &\leq (1-p)^{N/3} + o\left(n^{-20/3}\right)\\
        &\leq e^{-Np/3} + o\left(n^{-20/3}\right)\\
        &= o\left(n^{-20/3}\right).
    \end{align*}
    After a union bound over all vertices in $A$ and $B$, we conclude that, with probability $1-o(n^{-4})$, $a$ and $b$ are adjacent if and only if $\supp \chi(b) \subseteq \supp \chi(a)$.

    Learning the neighborhoods $N(T_i)$, $i = 1, \ldots, N$ with an optimal algorithm for CGT takes
    \begin{equation}\label{eqn: queries}
        O\left(\sqrt{|N(T_1)|} + \cdots + \sqrt{|N(T_N)|}\right)
    \end{equation}
    queries.
    By Cauchy-Schwarz, this is at most
    \[
    O\left(\sqrt{N}\right)\left(\sum_{i=1}^N|N(T_i)|\right)^{1/2}.
    \]
    Since $|N(T_i)|\leq m_{AB}$, the number of queries needed to learn all of the neighborhoods $N(T_i)$ is at most $O(N\sqrt{m_{AB}}) = \tilde{O}(d\sqrt{m_{AB}})$.
    Thus, the edges are learned with probability $1-o(n^{-4})$.
    
\end{proof}

\begin{remark}\label{rem: verts or edges}
    We note that we can bound the number of queries in (\ref{eqn: queries}) by $O(N\sqrt{|A|}) = \tilde{O}(d\sqrt{|A|})$.
    While this beats the bound of $\tilde{O}(d\sqrt{m_{AB}})$ when there are many edges, we will soon see that a bound in terms of the number of edges will prove to be more useful. 
\end{remark}

Like Angluin and Chen, we generalize from independent sets to sets with known edges.

\begin{lemma}\label{lem: improved general case}
    Let $G$ be a graph on $n$ vertices with maximum degree $d$.
    Suppose $A$ and $B$ are two disjoint, non-empty sets of vertices in $G$, inducing $m_A$ and $m_B$ known edges, respectively, and that there are $m_{AB}$ unknown edges between $A$ and $B$.
    Then there is a quantum algorithm that identifies the $m_{AB}$ crossing edges in $O(d\sqrt{m_{AB}}\log^5n)$ queries.
    The algorithm succeeds with probability $1-o(n^{-2})$.
\end{lemma}
\begin{proof}
    The idea here is to cover $A$ and $B$ with independent sets and then learn the edges that cross between them using Lemma \ref{lem: improved bipartite case}.
    To this end, set $N = 75 d \log n$ and $p = \frac{1}{2d}$.
    Let $A_1, \ldots, A_N$ and $B_1, \ldots, B_N$ be random subsets of $A$ and $B$, respectively, where each $a\in A$ ($b\in B$) independently appears in $A_i$ ($B_j$) with probability $p$.
    Again, the Chernoff bound tells us that each $a\in A$ and $b\in B$ appears in $\Theta(\log n)$ of the $A_i$'s and $B_j$'s with probability at least $1-o(n^{-2})$.
    
    For any $a\in A$, each of its at most $d$ neighbors has a $p = \frac{1}{2d}$ chance of appearing in $B_j$, so we don't expect the maximum degree in the subgraph induced by $A_i\cup B_j$ to be too large for any fixed $i$ and $j$.
    More concretely, set $t = \log n$ and let $X^{i,j}_{\geq t}$ be the number of vertices of degree at least $t$ in $G[A_i\cup B_j]$.
    By Markov's inequality, 
    \begin{align*}
        \Pr[\Delta(G[A_i\cup B_j]) \geq t] &\leq \E \left[X^{i,j}_{\geq t}\right]\\
        &= \sum_{v\in A\cup B}\Pr\bigg[\big|N[v]\cap(A_i\cup B_j)\big| \geq t\bigg]\\
        &\leq (|A|+|B|)\binom{d}{t}p^{t+1}\\
        &\leq (|A|+|B|)\left(\frac{edp}{t}\right)^tp.
    \end{align*}
    A union bound then shows that $\Delta(G[A_i\cup B_j]) \leq \log n$ simultaneously for all $i$ and $j$ with probability at least
    \[
    1 - N^2(|A|+|B|)\left(\frac{edp}{t}\right)^tp = 1-\frac{75^2}{2}(|A|+|B|)d\log^2n\left(\frac{e}{2\log n}\right)^{\log n}= 1-n^{-\Omega(\log \log n)}.
    \]

    Now a simple greedy algorithm can properly color the vertices of each $G[A_i]$ and $G[B_j]$ using at most $T = \log n+1$ colors (i.e., each color class forms an independent set).
    Note that this coloring requires no queries to our oracle since the edges of $G[A]$ and $G[B]$ are assumed to be known.
    We can then partition each $A_i$ and $B_j$ into $T$ (or perhaps fewer) independent sets $A_i = A_{i, 1}\cup \ldots \cup A_{i, T}$ and $B_{j, 1}\cup \ldots \cup B_{j,T}$.
    If we denote the number of edges that cross between $A_{i,k}$ and $B_{j,\ell}$ by $m_{i,k,j,\ell}$, then by Lemma \ref{lem: improved bipartite case}, we may learn these edges in at most $O(\sqrt{m_{i,k,j,\ell}}\log^2n)$ queries with probability at least $1-o(n^{-4})$.
    By the Cauchy-Schwarz inequality, the total number of queries required to learn the edges that cross between all pairs $A_{i,k}$ and $B_{j,\ell}$ is at most
    \[
    \sum_{i,j\leq N}\sum_{k,\ell\leq T}O(\sqrt{m_{i,k,j,\ell}}\log^2n) \leq O(\log^2n)\cdot NT\cdot \sqrt{\sum_{i,j\leq N}\sum_{k,\ell\leq T}m_{i,k,j,\ell}}
    \]
    and this succeeds with probability at least $1-o(n^{-2})$.
    As discussed before, each vertex appears in $\Theta(\log n)$ of the $A_i$ and $B_j$'s with high probability.
    Since the $A_{i,k}$'s ($B_{j,\ell}$'s) partition $A_i$ ($B_j$), the sum appearing under the square root above overcounts $m_{AB}$ by a factor of $O(\log^2n)$.
    With probability $1-o(n^{-2})$, the total number of queries is then
    \[
    O(NT\sqrt{m_{AB}}\log^3n) = O\left( d\sqrt{m_{AB}}\log^5n\right).
    \]
\end{proof}

\begin{remark}
    Montanaro and Shao prove a lemma similar to the one above, except that they obtain a result that holds with probability 1 at the cost of an additional factor of $d$ in the number of queries.
    This difference arises from their use of $\tilde{O}(d^2)$ subsets of $B$ that allow for a deterministic solution to the CGT problem.

    Per Remark \ref{rem: verts or edges}, we could have bounded the number of queries required to learn $E(A_{i,k}, B_{j,\ell})$ by $\tilde{O}(\sqrt{|A_{i,k}|})$.
    By Cauchy-Schwarz, this would yield a bound on the total number of queries of
    \[
    \sum_{i,j \leq N}\sum_{k, \ell \leq T}\tilde{O}\left(\sqrt{|A_{i,k}|}\right) = NT\cdot  \sum_{i\leq N}\sum_{k\leq T}\tilde{O}\left(\sqrt{|A_{i,k}|}\right) \leq \tilde{O}\left( (NT)^{3/2}\sqrt{|A|}  \right) = \tilde{O}\left( d^{3/2}\sqrt{|A|}\right).
    \]
    When $G[A\cup B]$ is $d$-regular, this matches the $\tilde{O}(d\sqrt{m_{AB}})$ bound from the above lemma.
    On the other hand, when $G[A\cup B]$ has few vertices of degree $d$, this gives a worse bound.
\end{remark}

\section{Learning a matching - Proof of Theorem \ref{thm: matching}}
In this section, we consider the case of learning a (not necessarily perfect) matching with quantum OR-queries.
We suppose that $G$ is a graph, promised to be a matching on $n$ vertices with $m \leq \lfloor n/2\rfloor$ edges.

\subsection{Choosing a random vertex partition}
Let $V_{1,1}, V_{1, 2}, \ldots, V_{1, T_1}$ be a partition of the vertex set of $G$ into $T_1 = \lfloor\sqrt m\rfloor$ parts, each part containing either $\lfloor n/T_1\rfloor$ or $\lceil n/T_1\rceil$ vertices, chosen uniformly at random from the set of all such partitions.
Any particular edge $e$ appears in $E(V_{1, j})$ with probability approximately $\frac{1}{T_1^2}$, so $\E[e(V_{1, j})] \approx \frac{m}{T_1^2} = 1$.
Furthermore, since $G$ is a matching, for any $k$ edges, $e_1, \ldots, e_k$, the events $\{e_i \in E(V_{1,j})\}$ are negatively correlated, so we may apply Theorem \ref{Chernoff} (and the remark following it) to obtain an upper tail estimate for $E(V_{1, j})$.
Specifically, for each $j$
\begin{equation}\label{edges}
\Pr\left[e(V_{1, j}) \geq  \log n\right] \leq \left(\frac{e^{\frac{T_1^2}{m}\log n-1}}{(\frac{T_1^2}{m}\log n)^{\frac{T_1^2}{m}\log n}}\right)^{m/T_1^2}  = n^{-\Omega(\log\log n)}.
\end{equation}
By a union bound over all $T_1$ parts, $e(V_{1, j}) < \log n$ simultaneously for all $1 \leq j \leq T_1$ with probability $1-n^{-\Omega(\log\log n)}$.

\subsection{Learning the edges within the parts}
Now we learn each of the edge sets $E(V_{1, j})$ using the classical algorithm Theorem \ref{thm: classical}.
Each part $V_{1,j}$ induces at most $\log n$ edges, with high probability, so we require at most $O(\log^2n)$ classical OR-queries to determine the edges in any $E(V_{1, j})$.
In total, $O(\sqrt m \log^2n)$ queries suffice to learn all of the $E(V_{1, j})$'s.

\subsection{Learning edges that cross between (some) parts}

Unlike in the case of internal edges, crossing edges are not negatively correlated, so we cannot apply Theorem \ref{Chernoff}.
Fortunately, we can still show concentration by simply estimating the distribution of the number of crossings.

If $k_1$ is a positive integer, the expression for $\Pr[e(V_{1, 2j-1}, V_{1, 2j}) \geq k_1]$ would be much simpler if we instead chose the $V_{1, j}$'s by individually sampling vertices independently (the disadvantage here is that the sizes of these subsets would then be random).
More specifically, suppose $\widetilde{V}_{1, 2j-1}$ is sampled by including each vertex independently with some probability $p_1$ and then $\widetilde{V}_{1, 2j}$ is sampled from $V\setminus \widetilde{V}_{1, 2j-1}$ by including each vertex independently with some other probability $p_1'$.
We can then think of crossings between $V_{1, 2j-1}$ and $V_{1, 2j}$ as crossings between $\widetilde{V}_{1, 2j-1}$ and $\widetilde{V}_{1, 2j}$ conditioned on these latter subsets being of the correct size ($n/T_1$).
The probability we are interested in is
\begin{equation}\label{prob}
    \Pr\big[e(V_{1, 2j-1}, V_{1, 2j}) \geq k_1\big] =
    \frac{\Pr\big[ e(\widetilde{V}_{1, 2j-1}, \widetilde{V}_{1, 2j}) \geq k_1,\ |\widetilde{V}_{1, 2j-1}| = |\widetilde{V}_{1, 2j}| = n/T_1 \big]}{\Pr\big[ |\widetilde{V}_{1, 2j-1}| = |\widetilde{V}_{1, 2j}| = n/T_1 \big]}
\end{equation}
Now we set $p_1 = 1/T_1$ and $p_1' = \frac{p_1}{1-p_1}$ so that 
\[
\E\big[|\widetilde{V}_{1, 2j-1}|\big] = \E\big[|\widetilde{V}_{1, 2j}| : |\widetilde{V}_{1, 2j-1}| = n/T_1\big] = n/T_1.
\]
With this choice and an application of Lemma \ref{lem: binomial}, the denominator of (\ref{prob}) is asymptotic to
\[
\frac{1}{\sqrt{2\pi np_1(1-p_1)}}\cdot \frac{1}{\sqrt{2\pi n(1-p_1)p_1'(1-p_1')}} = \frac{1}{2\pi np_1\sqrt{1-2p_1}}.
\]
We can upper bound the numerator in (\ref{prob}) by $\Pr\big[ e(\widetilde{V}_{1, 2j-1}, \widetilde{V}_{1, 2j}) \geq k_1\big]$, which is just
\[
2^{k_1}\binom{m}{k_1}p_1^{k_1}(1-p_1)^{k_1}(p_1')^{k_1} = 2^{k_1}\binom{m}{k_1}p_1^{2k_1}.
\]
To see this, first choose $k_1$ of the $m$ edges to cross.
For each chosen edge, we need one end to land in $\widetilde{V}_{1, 2j-1}$ (with probability $p_1$) and the other end to land in $\widetilde{V}_{1, 2j}$ (with probability $(1-p_1)p_1'$, conditioned on the first end landing in $\widetilde{V}_{1, 2j-1}$).
Now set $k_1 = \log n$.
Since $\binom{n}{k} \leq (\frac{en}{k})^k$ and $p_1 = 1/\sqrt m$, we have
\begin{equation}\label{crossing probability}
\begin{split}
\Pr\big[ e(\widetilde{V}_{1, 2j-1}, \widetilde{V}_{1, 2j}) \geq k_1\big] &\leq 2^{k_1} \left(\frac{emp_1^2}{k_1}\right)^{k_1}\\
& =  \left( \frac{2e}{\log n}\right)^{\log n}\\
&= n^{-\Omega(\log\log n)}.
\end{split}
\end{equation}


Returning to the $V_{i,j}$'s,
\begin{equation}\label{true crossing probability}
\Pr\big[e(V_{1, 2j-1}, V_{1, 2j}) \geq \log n\big] \leq n^{-\Omega(\log\log n)}\cdot \Theta\left(np_1\sqrt{1-2p_1}  \right)= n^{-\Omega(\log\log n)}.
\end{equation}

By Lemma \ref{lem: improved general case} (with $d = 1$), we can find all of the edges in $E(V_{1, 2j-1}, V_{1, 2j})$ in $O(\log^{11/2} n)$ quantum OR-queries with probability $1-o(n^{-2})$.
In total, it takes at most $O(\sqrt{m}\log^{11/2}n)$ queries to find all edges that cross between all of the matched sets and this succeeds with probability $1-o(T_1/n^2) = 1-o(n^{-3/2})$.

\subsection{Iterate}
We learn the remaining edges by merging the paired $V_{1, j}$'s and repeating the above step on the merged parts.
That is, let $V_{i, j} = V_{i-1, 2j-1}\cup V_{i-1, 2j}$ and set $T_i = \lceil T_{i-1}/2\rceil$.
Since the $V_{1, j}$'s were chosen uniformly at random, we may treat $\{V_{i, 1}, \ldots, V_{i, T_i}\}$ as an equipartition of $V$ into $\lceil n/T_i \rceil$ parts chosen uniformly at random from all such partitions.
Let $p_i = \frac{1}{T_i}$ be the (rough) probability that any particular vertex lands in $V_{i,j}$, so that $\E[e(V_{i, 2j-1}, V_{i, 2j})] \approx 2mp_i^2$.

In order for Lemma \ref{lem: improved general case} to be of any use here, we need to make sure that there are not too many crossing edges.
Here, we take ``too many'' to mean greater than $\log n$ times the expected number of crossings.
Accordingly, set $k_i = 2mp_i^2 \log n$.
Like in the previous section, we build $\widetilde{V}_{i, 2j-1}$ by selecting each vertex from $V$ independently with probability $p_i$ and then sample $\widetilde{V}_{i, 2j}$ from the remaining vertices in the same way, but with probability $p_i' = \frac{p_i}{1-p_i}$.
The corresponding version of equation (\ref{prob}) is
\[
    \Pr\big[e(V_{i, 2j-1}, V_{i, 2j}) \geq k_i\big] =
    \frac{\Pr\big[ e(\widetilde{V}_{i, 2j-1}, \widetilde{V}_{i, 2j}) \geq k_i,\ |\widetilde{V}_{i, 2j-1}| = |\widetilde{V}_{i, 2j}| = n/T_i \big]}{\Pr\big[ |\widetilde{V}_{i, 2j-1}| = |\widetilde{V}_{i, 2j}| = n/T_i \big]}
\]
Another application of Lemma \ref{lem: binomial} shows that the denominator of this expresssion is asymptotic to
\[
\frac{1}{2\pi np_i\sqrt{1-2p_i}}.
\]

Since $k_i = 2mp_i^2\log n$ and $mp_i^2 \geq 1$, we have
\begin{align*}
\Pr\big[ e(\widetilde{V}_{i, 2j-1}, \widetilde{V}_{i, 2j}) \geq k_i\big] &\leq 2^{k_i} \left(\frac{emp_i^2}{k_i}\right)^{k_i}\\
& \leq \left( \frac{e}{\log n}\right)^{2mp_i^2\log n}\\
&= n^{-\Omega(\log\log n)}.
\end{align*}
We then have that $\Pr\big[e(V_{i, 2j-1}, V_{i, 2j}) \geq k_i\big] = n^{-\Omega(\log\log n)}$

We upper bound the probability that there are ever too many crossings between the $V_{i,j}$ by
\begin{align*}
    \Pr\big[ e(V_{i, 2j-1}, V_{i, 2j}) \geq k_i\text{ for some }i,j\big] &\leq \sum_{i = 1}^{\lceil \log T_1\rceil}\sum_{j=1}^{\lfloor T_i/2\rfloor}\Pr\big[e(V_{i, 2j-1}, V_{i, 2j}) \geq k_i\big]\\
    &\leq \sum_{i=1}^{\lceil \log T_1\rceil}\frac{1}{2}T_i\cdot n^{-\Omega(\log\log n)}\cdot \Theta\left(np_i\sqrt{1-2p_i}\right)\\
    & \leq \lceil \frac{1}{2}\log m\rceil \cdot n^{-\Omega(\log\log n)}\\
    &= n^{-\Omega(\log\log n)}.
\end{align*}

Using Lemma \ref{lem: improved general case}, we can find all of the edges in $E(V_{i, 2j-1}, V_{i, 2j})$ in at most $O(\sqrt{k_i}\log^5n)$ queries, with probability $1-o(n^{-2})$.
The total number of queries it takes to learn all of these crossing edges is then at most
\[
\sum_{i=1}^{\lceil \log T_1\rceil} O\left(T_i\sqrt{k_i}\log^5n\right)
= O\left( \log^5n\right)\sum_{i= 1}^{\lceil \log T_1\rceil}T_i\sqrt{2mp_i^2\log n}.
\]
Recalling that $p_i = 1/T_i$ and $\log T_1 = O(\log m)$, this is equal to
\[
O\left( \sqrt{m}\log m \log^{11/2}n\right).
\]
This succeeds with probability
\[
1- \sum_{i=1}^{\lceil \log T_1\rceil}T_i \cdot o\left(n^{-2}\right) = 1- o\left(T_1/n^2\right) = 1- o(n^{-3/2}).
\]

\section{Learning a graph with bounded degree - proof of Theorem \ref{thm: max degree}}
We make slight modifications to the preceding algorithm in order to learn a graph with maximum degree $d$.
To show that the resulting algorithm, Algorithm \ref{alg: LearnGraph} terminates after a reasonable number of queries with high probability, we will again argue that not too many edges cross between randomly chosen subsets of vertices.
If we randomly select two subsets of vertices, the distribution of the number of crossing edges is more complicated than in the case of a matching.
However, by Vizing's theorem, Theorem \ref{Vizing}, we may decompose $G$ into $d$ or $d+1$ edge-disjoint matchings and then apply the concentration results we developed for matchings.
We stress that finding these matchings is not actually a step in our algorithm -- we simply use their existence to argue that not too many edges cross between random subsets.



Consider a decomposition of $G$ into edge-disjoint matchings, $E(G) = M_1 \cup M_1 \cup \cdots \cup M_{d+1}$ and write $m_i = |M_i|$.
We may take this decomposition to be equitable, so that $m_i \in \{\lfloor \frac{m}{d+1}\rfloor, \lceil \frac{m}{d+1}\rceil\}$ and $m_i\leq m_{i+1}$ for all $i<d+1$.
To see this, note that for distinct $i$ and $j$, $M_i\cup M_j$ is a collection of paths and cycles.
If $m_i < m_j$, then some component contains more edges from $M_j$ than from $M_i$.
We may then swap the edges between $M_i$ and $M_j$ in these unbalanced components until $|m_i-m_j|\leq 1$.
For any sets $A$ and $B$ and positive integer $t\leq d+1$, we let $e_t(A,B) = M_t \cap E(A,B)$ be the edges from $M_t$ that cross between $A$ and $B$.

\begin{algorithm}[h]
\caption{FindEdges$(G, m, d)$}
\label{alg: LearnGraph}
\textbf{Input}: a graph $G$ on $n$ vertices with unknown edge set that is promised to have $m$ edges and maximum degree $d$ \\
\textbf{Output}: the edge set $E(G)$
\begin{algorithmic}[1]
\State $T_1 \gets \lfloor\sqrt{\frac{m}{d+1}}\rfloor$
\State Randomly and equitably partition $V(G) = V_{1,1}\cup  \cdots \cup V_{1, T_1}$.

\For{$j = 1, \ldots, T_1$}
    \State Determine $E(V_{1, j})$ using the classical algorithm of Angluin and Chen, Theorem \ref{thm: classical}.
\EndFor

\For{$i = 1, \ldots, \lceil\log_2 T_1\rceil$}
    \For{$j = 1, \ldots, T_i/2$}
        \State Determine $E(V_{i, 2j-1}, V_{i, 2j})$ using Lemma \ref{lem: improved general case}.

        \State $V_{i+1, j} \gets V_{i, 2j-1}\cup V_{i, 2j}$
    \EndFor
    \State $T_{i+1} \gets \lceil T_i / 2\rceil$
\EndFor
\
\end{algorithmic}
\end{algorithm}

\subsection{Choosing a random (vertex) partition}
Like in the matching case, we begin by partitioning $V(G)$ into $T_1 = \sqrt m_1 = \Theta(\sqrt{m/d})$ (nearly) equally sized subsets, chosen uniformly at random from all such partitions, $V(G) = V_{1,1} \cup \cdots \cup V_{1, T_1}$.
For each $t\leq d+1$, the expected number of edges from matching $M_t$ that land in $V_{1, j}$ is approximately $m_t/T_1^2  \approx 1$ since the $m_t$'s differ from each other by at most one.
Calculation (\ref{edges}) carries over exactly, with $e_t(V_{1,j})$ in place of $e(V_{1,j})$ and with $m_t$ in place of $m$.
After union bounding over all $T_1$ parts and $d+1$ matchings, we obtain
\[
\Pr[e(V_{1,j}) \geq  (d+1)\log n\text{ for some }j] \leq (d+1)\cdot T_1\cdot n^{-\Omega(\log\log n)} = n^{-\Omega(\log\log n)}.
\]

\subsection{Learning the edges within the parts}
Exactly as in the case of a matching, with high probability, each part $V_{i, j}$ induces no more than $O(d\log n)$ edges of $G$.
We can use the classical algorithm of Theorem \ref{thm: classical} to learn the edges in $V_{1, j}$ in $O(d\log^2n)$ classical OR-queries for a total of $O(T_1\cdot d\log^2 n) = O(\sqrt{md} \log^2 n)$ queries to learn the edges induced by all of the $V_{1,j}$'s.

\subsection{Learning some of the crossings}
We again condition on a binomial sampling of the $V_{1, j}$'s to more easily estimate the distribution of the number of crossing edges.
For each $t\leq d+1$, calculations (\ref{prob}), (\ref{crossing probability}) and (\ref{true crossing probability}) carry over to $e_t(V_{1, 2j-1}, V_{1, 2j})$ and $e_t(\widetilde{V}_{1,2j-1}, \widetilde{V}_{1, 2j})$, with $m_t$ in place of $m$ and our new value for $T_1 = \sqrt{m_1}$.
We union bound over all $\frac{1}{2}T_1$ pairs of subsets and $d+1$ matchings to obtain
\[
\Pr[e(V_{1, 2j-1}, V_{1, 2j}) \geq (d+1) \log n \text{ for some }j] \leq (d+1)\cdot \frac{1}{2}T_1 \cdot n^{-\Omega(\log\log n)} = n^{-\Omega(\log\log n)}.
\]
By Lemma \ref{lem: improved general case}, with probability $1-o(n^{-2})$, we can learn all of the $E(V_{1, 2j-1}, V_{1, 2j})$'s in at most
\[
O\left(d\sqrt{m}\log^{11/2}n\right)
\]
quantum OR-queries.

\subsection{Iterate}
Since the number of crossings for each matching concentrates, so does the number of crossings for $G$.
That is, for each $1 < i \leq \lceil \log T_1\rceil$ and $j< i/2$, let $V_{i,j} = V_{i-1, 2j-1}\cup V_{i-1, 2j}$ and set $T_i = \lceil T_{i-1}/2\rceil$.
For every $1\leq i \leq \lceil \log T_1\rceil$, $j\leq \frac{1}{2}i$ and $t\leq d+1$, the expected number of edges from $M_t$ that cross between $V_{i, 2j-1}$ and $V_{i, 2j}$ is approximately $2m_tp_i^2$, where $p_i = 1/T_i$.
Calculations (\ref{prob}), (\ref{crossing probability}) and (\ref{true crossing probability}) carry over and we have
\[
e_t(V_{i, 2j-1}, V_{i, 2j}) \leq 2m_tp_i^2\log n
\]
with probability $1-n^{-\Omega(\log\log n)}$.
Union bounding over all values of $i$, $j$ and $t$, we have
\[
e(V_{i, 2j-1}, V_{i,2j}) =  O\left( dm_1p_i^2\log n\right) = O\left(m p_i^2\log n\right)
\]
for all $i$, $j$ and $t$ with probability $1-n^{-\Omega(\log\log n)}$.
By Lemma \ref{lem: improved general case}, the total number of queries we use to learn the crossing edges is at most
\[
\sum_{i=1}^{\lceil \log T_1\rceil}O\left( T_i\cdot d\sqrt{2mp_i^2\log n}\log^5n \right) = O\left( d\sqrt m\log(m/d)\log^{11/2}n \right),
\]
and the probability that this procedure succeeds is
\[
1 - \sum_{t=1}^{d+1}\sum_{i=1}^{\lceil \log T_1\rceil}T_i \cdot o\left(n^{-2}\right) = 1-(d+1)T_1\cdot o\left(n^{-2}\right) = 1 - o\left(\sqrt{dm}\cdot  n^{-2}\right) = 1-o\left(n^{-1/2}\right).
\]

\section{Concluding remarks}
In the case where $G$ consists of two cliques on $n/2$ vertices connected with arbitrary edges, the only queries that give any information are of the form $\{u,v\}$, where $u$ and $v$ come from opposite cliques.
By the result of Beals et al. \cite{beals2001quantum} on the lower bound for the quantum query complexity of learning the parity of a string, a quantum algorithm requires $\Omega(n^2) = \Omega(m)$ queries to learn the edge set of this graph in the worst case.
In this case, the maximum degree $d$ of $G$ is $\Theta(n)$, so $d\sqrt{m} = \Theta(n^2) = \Theta(m)$ and our algorithm performs optimally, up to logarithmic factors, as does the algorithm of Montanaro and Shao, mentioned in the first part of Theorem \ref{thm: montanaro shao}.
Our algorithm outperforms Montanaro and Shao's $O(m\log(\sqrt m \log n) + \sqrt{m}\log n)$ algorithm in the case where $d = o(\sqrt m)$ (ignoring logarithmic factors).

What can we say when $G$ is promised to belong to some specific class of graphs?
For example, when $G$ has maximum degree $d = O(1)$, our algorithm learns $G$ in $\tilde{O}(\sqrt m)$ queries.
When $G$ is promised to be a star with $m$ edges, the general algorithms of Theorem \ref{thm: montanaro shao} and Algorithm \ref{alg: LearnGraph} learn $G$ in $\tilde{O}(m)$ and $\tilde{O}(m^{3/2})$ quantum queries, respectively, while Montanaro and Shao give an ad-hoc algorithm (\cite{montanaroShao} Proposition 9) that learns $G$ in $O(\sqrt m)$ queries.
Furthermore, if $G$ is promised to be a clique on $k$ vertices (the rest of the vertices remaining isolated), then Theorem \ref{thm: montanaro shao} and Algorithm \ref{alg: LearnGraph} both learn $G$ in $\tilde{O}(k^2)$ queries, while an ad-hoc algorithm (\cite{montanaroShao} Proposition 8) can accomplish this in just $O(\sqrt k)$ queries.
Can we say anything interesting if $G$ is promised to be a tree?

\bibliographystyle{abbrv}
\bibliography{hamiltonian}

\end{document}